\documentclass[11pt]{amsart}
\usepackage{amssymb}
\usepackage{amsmath}
\usepackage{algorithm,algorithmic}
\usepackage[bookmarks]{hyperref}
\usepackage{setspace}

\newtheorem{theorem}{Theorem}

\newtheorem{definition}{Definition}
\newtheorem{corollary}{Corollary}
\newcommand{\argmax}{\operatornamewithlimits{arg\ max}}

\begin{document}
\onehalfspacing
\title{Paging with dynamic memory capacity}
\author{Enoch Peserico}
\thanks{This work was supported in part by Univ. Padova under Strategic Project AACSE}

\maketitle
\begin{abstract}
We study a generalization of the classic paging problem that allows the amount of available memory to vary over time -- capturing a fundamental property of many modern computing realities, from cloud computing to multi-core and energy-optimized processors. It turns out that good performance in the ``classic'' case provides no performance guarantees when memory capacity fluctuates: roughly speaking, moving from static to dynamic capacity can mean the difference between optimality within a factor $2$ in space and time, and suboptimality by an arbitrarily large factor. More precisely, adopting the competitive analysis framework, we show that some online paging algorithms, despite having an optimal $(h,k)-$competitive ratio when capacity remains constant, are not $(3,k)-$competitive for any arbitrarily large $k$ in the presence of minimal capacity fluctuations.

In this light it is surprising that several classic paging algorithms perform remarkably well even if memory capacity changes adversarially - even without taking those changes into explicit account! In particular, we prove that LFD still achieves the minimum number of faults, and that several classic online algorithms such as LRU have a ``dynamic'' $(h,k)-$competitive ratio that is the best one can achieve without knowledge of future page requests, even if one had perfect knowledge of future capacity fluctuations (an exact characterization of this ratio shows it is almost, albeit not quite, equal to the ``classic'' ratio $\frac{k}{k-h+1}$). In other words, with careful management, knowing/predicting future memory resources appears far less crucial to performance than knowing/predicting future data accesses.
\end{abstract}

\newpage
\vspace*{3mm}
\section{Introduction}
\label{sec:intro}
\vspace*{1mm}

This article examines a generalization of the classic paging problem that allows the amount of available memory to vary over time. After briefly reviewing the paging problem (subsection \ref{sub:paging}) this section motivates paging with dynamic capacity (subsection \ref{sub:elastic}) and provides an overview of our results and of the organization of the rest of the article (subsection \ref{sub:results}).

\subsection{The paging problem}
\label{sub:paging}
The memory/data storage system of modern computing devices is almost always organized as a hierarchy of several layers of progressively larger capacity but also higher access cost (in terms of both time and energy). The ever widening gap in both capacity and cost between different layers makes the \emph{paging problem}, i.e. the problem of efficiently orchestrating the flow of information across the memory hierarchy, crucial to the performance of a computing device. The most widely used theoretical model for studying paging is that of a two-layer system: a smaller \emph{memory} layer with a capacity of $k$ \emph{pages} (data blocks), and a larger layer of infinite capacity whose pages can only be accessed by first copying them into memory -- an operation usually termed a {\em (page) fault}. Given any sequence of pages that must be accessed in order, an algorithm for the paging problem must choose which page(s) to ``evict'' from memory, whenever a new page must be copied into it, so as to minimize the total number of faults.

The simple algorithm LFD (Longest Forward Distance) that evicts the page accessed furthest in the future has long been known to be optimal \cite{belady}. However, paging is often studied as an {\em online} problem, i.e. an algorithm can decide evictions only on the basis of past requests. One very popular framework for evaluating the performance of online paging algorithms is that of \emph{competitive analysis}~\cite{competitive}. A paging algorithm is said to have an $(h,k)-$\emph{competitive ratio} of (no more than) $\rho$ if, for every request sequence, it incurs in expectation with a memory of capacity $k$ at most $\rho$ times as many faults as an optimal offline algorithm incurs with a memory of capacity $h\leq k$, plus a number of faults independent of the request sequence. The ratio $\frac{k}{h}$ is called the \emph{resource augmentation}. Resource augmentation and competitive ratio capture, respectively, the space and access cost overheads incurred by an online algorithm.

Many simple, deterministic algorithms including LRU\footnote{Least Recently Used -- evict the least recently accessed page}, FIFO\footnote{First In First Out -- evict the page brought least recently into memory.}, FWF\footnote{Flush When Full -- evict all pages whenever memory is full and space is needed.} and CLOCK\footnote{Mark any page accessed; to evict a page, cycle through pages, unmarking those found marked and evicting the first found unmarked.} have an $(h,k)$-competitive ratio of $\frac{k}{k-h+1}$ \cite{lrufifo,online}; and the same ratio holds for RAND\footnote{Evict a page chosen uniformly at random.} \cite{rand}. This ratio is optimal for deterministic algorithms, and even for randomized ones if page requests can depend on previous choices of the paging algorithm (the ``adaptive adversary'' model~\cite{rand} which we adopt throughout this article\footnote{More precisely, in the \emph{online} adaptive adversary model, the choices of the reference offline algorithm can depend only on the \emph{past} random choices of the online algorithm; in the \emph{offline} adaptive adversary model they can depend on the random choices of the online algorithm over the entire request sequence. The bounds above -- and in fact all the bounds in this article -- hold for both models, with one exception: the upper bound on the competitive ratio of RAND above, and the corresponding upper bound we provide for RAND in theorem~\ref{thm:elasticrand}, only hold in the adaptive online model.}).
 
Since $\frac{k}{k-k/2+1}<2$, \emph{many simple online algorithms never fare worse than the optimal offline algorithm would on a memory system with half the capacity and twice the access cost}. This justifies the use of competitive analysis for preliminary performance evaluation of paging algorithms. Its ``worst-case'' approach may be somewhat pessimistic, but it is not overly so for many popular online paging algorithms -- for which it provides guarantees of performance within a factor $2$ of the optimal \emph{under any workload} (in terms of faults and required memory capacity). In contrast, the finer granularity evaluation provided by experimental benchmarking is inevitably tied to specific workloads.

\cite{online}, \cite{onlinesurvey05} and~\cite{competitivealternatives} provide three excellent surveys of the many variants of competitive analysis for the paging problem: these include somehow limiting the choice of the adversarial request sequence
~\cite{marking,diffuseadversary,accessgraph,multifinger},
amortizing the performance evaluation over a spectrum of sequences~\cite{maxmax,averageorder,bijective} 
or of memory capacities~\cite{youngcachechanges}, considering pages of different size and access cost~\cite{iranicachechanges,youngcachechanges}, and accounting for the non-zero cost of non-fault requests~\cite{nonfault}. 

\subsection{Paging with dynamic capacity}
\label{sub:elastic}

Throughout the long history and the many variants of the paging problem, memory capacity has generally been assumed to remain fixed throughout the request sequence. This no longer reflects many important computing realities.

In a cloud computing environment, the amount of physical memory available to an individual virtual machine does vary considerably over time depending on the virtual machine's load and on the number, load and relative class of service of other virtual machines hosted on the same hardware. Even on a simple PC, most modern operating systems have and use the option of declaring some critical virtual pages temporarily ``unswappable'', pinning them in main memory and thus reducing the amount of main memory available to user processes. 

Memory fluctuations also take place when considering the cache-RAM interface -- in which case memory represents cache memory and pages represent cache lines. In many multi-core processor designs cache capacity is partitioned \emph{dynamically} between different cores \cite{utilitypartition}. And low-power chip designs can often dynamically disable underutilized portions of the cache to save energy \cite{sleepycache}, again resulting in a capacity that can vary over time.

This article studies an extension of the classic paging problem that addresses these issues allowing memory capacity to fluctuate between $1$ and $k$ pages (instead of being constantly equal to $k$). These fluctuations may be either known beforehand to the paging algorithm (an ``offline'' problem), or unknown until they take place (an ``online'' problem). Note that, although there exists a large body of work on servicing the same request sequence with a policy that is simultaneously ``good'' on memories of different~\cite{youngcachechanges} and perhaps unknown~\cite{cacheoblivious} \emph{but unchanging} capacity, allowing capacity to vary dynamically during the course of the computation is an entirely different problem; as we shall see, a solution to the former does not guarantee even an approximate solution to the latter.

\cite{ramrental0,ramrental1} recently introduced the related, but fundamentally different, problem of \emph{RAM rental}. In RAM rental memory capacity can fluctuate \emph{under control of the paging algorithm}, and the goal is to minimize a linear combination of average capacity and fault rate over time. In practice there are usually very strong constraints on the set of admissible capacity values, on how they can change over time, and on their relative costs (which may themselves fluctuate). Also, a number of architectural approaches (e.g.~\cite{exokernel}) decouple the portion of the system responsible for page replacement from that responsible for capacity allocation. Then, assuming as we do that capacity fluctuations are not controlled by the paging algorithm (in fact, that they may be unknown beforehand and even chosen adversarially) leads to a more robust evaluation of page replacement policies. As we shall see, it turns out that page replacement \emph{can} be decoupled well from capacity control, yielding robust replacement policies that do not depend on the capacity choices or costs, while at the same time simplifying the RAM rental problem (which is still widely open -- e.g. little is known in terms of lower bounds or resource augmentation).

\subsection{Our results}
\label{sub:results}

The rest of this article is organized as follows. 
Section~\ref{sec:formal} introduces some formalism and terminology. In particular, it extends the notion of ``online vs. offline'' problem to encompass the extra dimension of future memory capacity, and it extends the notion of $(h,k)-$resource augmentation to the dynamic capacity scenario (in a nutshell, restricting the offline algorithm to at most a fraction $\frac{h}{k}$ of the online algorithm's \emph{current} memory capacity).

Section~\ref{sec:bad} shows the existence of online paging algorithms that have an (optimal) $(h,k)$-competitive ratio of $\frac{k}{k-h+1}$ in the ``classic'' paging model, and yet are no longer $(3,k)-$competitive \emph{for any arbitrarily large $k$} if their memory capacity is subject to single page fluctuations. This very negative result provides strong justification for our inquiry, as one cannot infer performance in the presence of (even minimal) memory fluctuations from performance in their absence.

In this light, it is quite surprising that many well-known algorithms perform remarkably well in the presence of memory fluctuations \emph{even if those fluctuations are chosen adversarially}. In section~\ref{sec:good} we show that the classic LFD algorithm remains optimal for all possible memory capacity fluctuations even though it does not explicitly take those fluctuations into account (i.e. it is an online algorithm in terms of memory fluctuations). We also show that in the dynamic capacity framework every online algorithm that is either marking \cite{marking} like LRU, FWF or MARK\footnote{Mark any page accessed; evict a random unmarked page, first unmarking all pages if all are marked.}, or \emph{dynamically conservative} (a simple refinement of the notion of ``conservative algorithm''\cite{youngconservative}), like LRU, CLOCK or FIFO, has an $(h,k)-$competitive ratio no larger than 
$\rho_{EL}(h,k)= \max_{k'\leq k, k'\in \mathcal{N}} 
\frac{k'}{k'-\lfloor h \frac{k'}{k} - \frac{h}{k}\rfloor}$.
Exactly the same bound holds for RAND (against an online adaptive adversary). 

Section~\ref{sec:tight} analyses $\rho_{EL}(h,k)$. We show that it is a lower bound to the $(h,k)-$competitive ratio achievable by any online paging algorithm in the presence of memory fluctuations, proving the optimality of marking and dynamically conservative algorithms. We also show that $\rho_{EL}(h,k)$ almost, but not quite, matches the ``classic'' bound of $\frac{k}{k-h+1}$ on the $(h,k)-$competitive ratio. More precisely, $\rho_{EL}(h,k)$ is at least $1+(\frac{1}{k}-\frac{2}{k^2})$ times as large as $\frac{k}{k-h+1}$ for any odd $h$ and $k=2h$, but it is always less than $1+\frac{1}{k}$ times as large -- and if $h>k-\sqrt{k}$ the two quantities actually coincide.

Section~\ref{sec:applications} briefly looks at the implications of our results for the RAM rental problem. In a nutshell, since many simple replacement are near optimal regardless of capacity fluctuations, RAM rental is simplified into the problem of just choosing a ``good'' capacity sequence without worrying about replacement. 

Finally, section~\ref{sec:end} summarizes our results and looks at their significance and at possible directions of future work.

\vspace*{3mm}
\section{Some formalism/terminology}
\label{sec:formal}
\vspace*{1mm}

We can easily extend the notion of request sequence $\sigma=r_1,r_2,\dots$ to the case of memory fluctuations. We simply assume that, interleaved with standard page requests, it is possible to have two additional types of requests, \emph{growths} and \emph{shrinks}. On a growth, memory capacity increases by $1$ page; on a shrink, it decreases by $1$ -- and if the memory was full a page must be evicted. We assume that initially memory capacity is $0$. Throughout the rest of this article, we denote a growth request by the symbol $+$ and a shrink request by the symbol $-$, and we denote $k$ consecutive growths / shrinks by $+^k$ and $-^k$. Thus, a standard request sequence $p_1,\dots,p_n$ on a memory of capacity $k$ simply becomes $+^k, p_1,\dots, p_n$ in the more general dynamic capacity framework. 

The request sequence automatically induces a \emph{page sequence} $\pi=<p_1, p_2, \dots>$ (the sequence of requested pages $p_1,p_2,\dots$, as in the classic paging problem) and a \emph{capacity sequence} $\mu=m_1,m_2,\dots$ where $m_i$ is the memory capacity immediately before the request for $p_i$ (i.e. it is equal to the number of $+$s minus the number of $-$s in the request prefix ending with $p_i$).
Note that the presence of growths and shrinks introduces a second aspect of ``onlineness''. More formally:

\begin{definition}
\label{def:elastic_online}
A paging algorithm ALG is online relative to the page sequence if its eviction choices before servicing a request are independent of any future page requests; otherwise it is offline relative to the page sequence. Similarly, ALG is online relative to the capacity sequence if its eviction choices before servicing a request are independent of any subsequent growths and shrinks; otherwise it is offline relative to the capacity sequence. ALG is a fully online, partially offline and fully offline paging algorithm if it is online relative to (respectively) both, one, or neither of the page and the capacity sequence.
\end{definition}

Thus, in the dynamic capacity model, all well-known paging algorithms such as LRU, FIFO, FWF, CLOCK, RAND and MARK are fully online, and LFD is partially offline, being offline relative to the page sequence even though it is online relative to the capacity sequence.

We can easily extend the notion of $(h,k)-$competitive ratio to the dynamic capacity model by comparing the cost (i.e. number of faults) incurred by an online algorithm whose memory capacity never exceeds $k$ to the cost incurred by an offline algorithm whose memory capacity never exceeds $\frac{h}{k}$ times that of the online algorithm. More formally, denote by OPT the optimal offline algorithm, and by $c_{ALG}(\pi, \mu)$ the cost incurred by an algorithm ALG when servicing a page sequence $\pi=p_1, \dots, p_n$ with a capacity sequence $\mu = m_1, \dots, m_n$. Also, given a capacity sequence $\mu = m_1, \dots, m_n$ and a non-negative number $a$, denote by $\lfloor a\cdot\mu\rfloor$ the capacity sequence $m'_1,\dots, m'_n$ with $m'_i = \lfloor a\cdot m_i\rfloor$. Then:

\begin{definition}
\label{def:elastic}
A paging algorithm ALG has a \emph{dynamic} $(h,k)-$competitive ratio of (at most) $\rho$ if there exists some constant $d$ such that, for any page sequence $\pi = p_1, \dots, p_n$ and any capacity sequence $\mu = m_1, \dots, m_n$ such that, $\forall i$, $m_i\leq k$:
\begin{equation*}
c_{ALG}(\pi,\mu) \leq \rho\cdot c_{OPT}(\pi,\lfloor\frac{h}{k}\cdot\mu\rfloor) + d
\end{equation*}
\end{definition}

Note that the dynamic $(h,k)-$competitive ratio of an algorithm is always an upper bound to its $(h,k)-$competitive ratio. Thus online paging with dynamic capacity is in some sense ``harder'' than classic online paging, and no online algorithm can have a dynamic $(h,k)-$competitive ratio lower than the ``classic'' ratio $\frac{k}{k-h+1}$.

\vspace*{3mm}
\section{Minimal capacity fluctuations can lead to arbitrarily large performance degradation}
\label{sec:bad}
\vspace*{1mm}

This section shows that there exist online paging algorithms that do not depend explicitly on memory capacity, and that have an optimal $(h,k)$-competitive ratio in the classic setting of fixed memory capacity, but are not competitive at all, even with arbitrary resource augmentation, when faced with even slight fluctuations in memory capacity. Consider the online paging algorithm LFRU (Least Frequently / Recently Used) that starts as LRU and then alternates between LFU and LRU -- switching from LRU to LFU after any palindrome subsequence incurring more faults in its second half, and switching from LFU to LRU after any palindrome subsequence incurring more faults in its first half:

\begin{algorithm}
\caption{LFRU: {s}ervice $p_0,\dots,p_n$ as follows}
\label{alg:lfru}
\small
\begin{algorithmic}
\STATE at $p_0$ POLICY $\leftarrow$ LRU
\STATE
\FOR{$i=1\dots n$}
\STATE
\STATE {\bf if}  at $p_i$ POLICY = LRU AND $\exists j<i$:\\
$<p_j\dots p_i>$ is palindrome
AND faults($p_j\dots p_{\lfloor\frac{i+j}{2}\rfloor}$) $<$ faults($p_{\lceil\frac{i+j}{2}\rceil}\dots p_i$)
\STATE {\bf then} at $p_{i+1}$ POLICY $\leftarrow$ LFU
\STATE
\STATE {\bf else if}  at $p_i$ POLICY = LFU AND $\exists j<i$:\\
$<p_j\dots p_i>$ is palindrome
AND faults($p_j\dots p_{\lfloor\frac{i+j}{2}\rfloor}$) $>$ faults($p_{\lceil\frac{i+j}{2}\rceil}\dots p_i$)
\STATE {\bf then} at $p_{i+1}$ POLICY $\leftarrow$ LRU
\STATE
\ENDFOR
\end{algorithmic}
\end{algorithm}

We would convince the reader that LFRU, while somewhat artificial and difficult to implement in practice, is not too different from many real-world paging heuristics designed for static memory capacity (note that the behaviour of LFRU, like that of LRU and LFU, does not depend explicitly on memory capacity). In fact, pure LRU tends to be outperformed in practice by various LRU/LFU hybrids~\cite{lrfu,arc}.
The reason for this is the regrettably common coexistence of ``local'' or ``temporal'' computations, exhibiting a high degree of temporal locality and data reuse, with ``streaming'' computations that access long sequences of sequential data with no temporal locality at all. In such cases, under LRU and similar policies such as CLOCK, streaming data not only gain no benefit from the use of a cache (since every new access is a fault) but \emph{pollute} the cache, forcing the eviction of temporal data and preventing the temporal computation from deriving more than a minimal benefit from the cache. One possible solution is to combine LRU with eviction schemes biased, like LFU, against data that have no reuse history even if their last (and only) access was very recent. And since LRU performs best when future requests are a ``mirror image'' of the past, it may seem reasonable to switch to it when such palindrome sequences exhibit good caching behaviour, and switch to LFU when such palindrome sequences exhibit poor caching behaviour -- which is exactly what LFRU does.

It turns out that LFRU has an optimal $(h,k)-$competitive ratio in the classic paging model where memory capacity is fixed. At the same time, even if faced with capacity fluctuations of just a single page, and even if allowed the use of an arbitrarily large amount of memory, LFRU's fault rate can be arbitrarily larger than that of an offline algorithm running with just $3$ pages of memory.
More formally we prove:

\begin{theorem}
\label{thm:LFRU}
LFRU has an $(h,k)-$competitive ratio equal to $\frac{k}{k-h+1}$ if memory capacity is constant, but has no finite dynamic $(h,k)-$competitive ratio for any $h\geq 3$ and any arbitrarily large $k$.
\end{theorem}

\begin{proof}
Let us first prove that LFRU has an $(h,k)-$competitive ratio equal to $\frac{k}{k-h+1}$ if memory maintains an arbitrary but fixed capacity $k$. We need only prove that, as long as LFRU keeps behaving as LRU, on no page request sequence a palindrome subsequence incurs more faults in its second half: then LFRU keeps behaving exactly as LRU and shares its $(h,k)$-competitive ratio of $\frac{k}{k-h+1}$. 

Consider a palindrome page subsequence  $\pi=p_{i_1},\dots,p_{i_\ell},p_{i_\ell},\dots,p_{i_1}$ of even length $2\ell$, containing $\lambda\leq\ell$ distinct pages $p_1,\dots,p_\lambda$. Note that, if $\lambda\leq k$, by the end of the first half of $\pi$ the $\lambda$ most recently requested pages are $p_i,\dots,p_\lambda$, which are then in memory and prevent any fault from taking place during the second half of $\pi$. Then, we need only consider the case $\lambda>k$.

Let us focus on the first half of $\pi$. For each distinct page, we analyse separately the first request to it (which we call a \emph{cold} request), and the remaining requests, if any (which we call \emph{hot} requests). The $i^{th}$ cold request is certainly a fault for any $i>k$, since at least $k$ distinct pages have been requested before it in $p_{i_1},\dots,p_{i_\ell}$; so the number of cold requests incurring faults is at least $\ell-k$. Let us now look at hot requests, and let $r_i$ be the number of hot requests of $p_i$ in the first half of $\pi$. For $1\leq i\leq \lambda$ and $1\leq j \leq r_i$, let $D_i^j$ be the set of distinct pages requested between the $j^{th}$ hot request for $p_i$ and the previous request for $p_i$, inclusive (so $D_i^j$ always includes $p_i$). Then the $j^{th}$ hot request to $p_i$ is a fault if and only if $|D_i^j|>k$, and the total number of faults in $p_{i_1},\dots,p_{i_\ell}$ is:
\begin{equation}
\label{eqn:firsthalf}
f_\pi^{\frac{1}{2}} \geq (\ell-k) + |\{(i,j):|D_i^j|>k\}|
\end{equation}

Let us now focus on the second half of $\pi$. Again, we divide requests for any distinct page into a cold request (the first) and hot requests (subsequent ones, if any). The first $k$ cold requests of $p_{i_\ell},\dots,p_{i_1}$ are for the last $k$ distinct pages requested in $p_{i_1},\dots,p_{i_\ell}$, which are then present in memory at the beginning of $p_{i_\ell},\dots,p_{i_1}$. So in $p_{i_\ell},\dots,p_{i_1}$ none of the first $k$ cold requests incurs a fault, yielding and at most $\lambda-k$ faults on cold requests. Let us now look at the hot requests of $p_{i_\ell},\dots,p_{i_1}$; those for $p_i$ are obviously $r_i$, as in the first half of $\pi$. For $1\leq i\leq \lambda$ and $1\leq j \leq r_i$, let $\bar{D}^i_j$ be set of distinct pages between the $j^{th}$ hot request for $p_i$ and its previous request, including $p_i$ itself; then the $j^{th}$ hot request for $p_i$ is a fault if and only if $|\bar{D}_i^j|>k$, and the total number of faults in $p_{i_\ell},\dots,p_{i_1}$ is:

\begin{equation}
\label{eqn:secondhalf}
\bar{f}_\pi^{\frac{1}{2}} \leq (\ell-k) + |\{(i,j):|\bar{D}_i^j|>k\}|
\end{equation}

It is crucial to observe that, since $\pi$ is palindrome, $\bar{D}^i_j=D^i_{r_i-j+1}$. Then $|\{(i,j):|\bar{D}_i^j|>k\}|=|\{(i,j):|D_i^j|>k\}|$ and $f_\pi^{\frac{1}{2}}\geq\bar{f}_\pi^{\frac{1}{2}}$. The analysis is virtually identical for palindrome subsequences of odd length; and thus with static memory capacity LFRU incurs no more faults on the second half of any palindrome subsequence than in the first half and has an $(h,k)-$competitive ratio equal to $\frac{k}{k-h+1}$.

To prove that LFRU can incur arbitrarily more faults than an optimal offline algorithm OPT when memory capacity fluctuates -- even if OPT is limited to a capacity fluctuating between capacity $3$ and $2$, while LFRU's fluctuates between $3m$ and $3m-1$ for an arbitrarily large $m$ -- we show how LFRU can be coaxed into, and kept in, LFU behaviour, and how that behaviour can result in arbitrarily more faults than OPT even with arbitrarily larger capacity.

Denote by $p^\ell$ the subsequence consisting of $\ell$ consecutive requests for $p$, and consider the page sequence: 
\begin{equation*}
\pi_1=<p_1,p_2,p_3^\ell\dots,p_{3m}^\ell,p_1,p_2,p_3,\dots,p_{3m},p_{3m},\dots,p_2,p_1>
\end{equation*}
with LFRU's memory capacity fixed at $3m$ except for the last $3m$ requests, during which it drops by $1$ to $3m-1$. It is immediate to see that when capacity drops $p_1$ is evicted, and that the last $6m$ requests form a palindrome subsequence experiencing a fault (only) on the last request. Thus, on the last request of $\pi$, LFRU switches to LFU behaviour, and evicts $p_2$ (which, like $p_1$, has experienced $\ell-1$ fewer requests than every other page $p_i$, $i\geq 3$). If the subsequence of requests $\pi'=(p_2,p_1)^{\ell-1}$ follows, with memory capacity remaining fixed at $3m-1$, LFRU keeps evicting in turn $p_1$ and $p_2$, which remain the two pages having experienced the fewest requests; and LFRU incurs at least $2\ell-2$ faults.

An optimal algorithm (or even just LRU) with memory capacity $3$ throughout all but the last $3m$ requests of $\pi$, and with memory capacity $2$ thereafter, would instead incur no more than $3m+3m+3m=9m$ faults during $\pi$, and no faults at all during $\pi'$ (retaining only $p_1$ and $p_2$ in memory). Thus, since $\ell$ can be chosen arbitrarily larger than $m$, LFRU does not have a finite $(3,3m)-$competitive ratio for any arbitrarily large $m$.
\end{proof}

\vspace*{3mm}
\section{Even adversarial fluctuations can be addressed efficiently (and ``implicitly'')}
\label{sec:good}
\vspace*{1mm}

In the light of theorem \ref{thm:LFRU} it may be somewhat surprising many well-known ``good'' paging algorithms still perform remarkably well in the dynamic capacity setting -- even though they do not take memory fluctuations into explicit account. It is very easy to prove:

\begin{theorem}
\label{thm:elasticlfd}
LFD incurs the minimal number of faults on any request sequence.
\end{theorem}

\begin{proof}
We can safely ignore algorithms leaving unoccupied space in memory after an eviction, as such an eviction could be delayed without incurring additional faults. Let a page be \emph{close} if it will be accessed before another page currently in memory, \emph{far} otherwise. LFD is the algorithm evicting no close pages. We prove the theorem showing that one can always eliminate the earliest close eviction without altering previous evictions or increasing the number of faults.

Let $p$ be the close page evicted earliest, at time $t$, by an algorithm $ALG$ servicing a request sequence. Consider the algorithm $\overline{ALG}$ that operates as $ALG$ until $t$, when it instead evicts a far page $\overline p$, and then operates as follows. Denote by $M$ and $\overline{M}$ the sets of pages respectively in $ALG$'s and $\overline{ALG}$'s memory. When both $ALG$ and $\overline{ALG}$ must incur an eviction, $\overline{ALG}$ evicts the same page as $ALG$ if possible; otherwise, when $\overline{ALG}$ must incur an eviction, it evicts a page not in $M$ (as soon as $\overline{M}=M$, $ALG$ and $\overline{ALG}$ coincide). After $t$, let $t'$ be the time of the first request or eviction of either $p$ or $\overline{p}$. Until $t'$ ALG and $\overline{ALG}$ incur exactly the same faults and evictions, and thus $M\setminus \overline{M}=\{\overline{p}\}$ and $\overline{M}\setminus M=\{p\}$. At $t'$ $\overline{ALG}$ evicts $p$ if and only if $ALG$ evicts $\overline{p}$ -- in which case $M$ and $\overline{M}$ converge. Otherwise $p$ is requested at $t'$ and $ALG$, but not $\overline{ALG}$, incurs a fault; and since $\overline{ALG}$ never evicts a page unless $ALG$ also has evicted it, $|\overline{M}\setminus M|$ never increases after $t$, and drops to $0$ no later than the first fault incurred by $\overline{ALG}$ and not by $ALG$. In both cases $\overline{ALG}$ incurs no more misses than $ALG$.
\end{proof}

It is interesting to note that theorem~\ref{thm:elasticlfd} yields as an immediate corollary theorem~$4.1$ in~\cite{multicore} -- in a nutshell, for a given, dynamically changing, partition of the memory space between different processes, using LFD for each process on its own partition yields the minimum \emph{total} number of faults. It is not, however, immediately obvious that the result in~\cite{multicore} implies our theorem~\ref{thm:elasticlfd}. Furthermore, the result in~\cite{multicore} is only stated, and not proved -- the proof is deferred to the full version of the article because of its complexity compared to the ``classic'' proof of LFD's optimality. 
\newline\newline
Let us now focus on \emph{online} paging algorithms. It turns out that the dynamic $(h,k)-$competitive ratio achievable by many well-known online algorithms is almost, but not quite, as good as the ``plain'' $(h,k)-$competitive ratio $\frac{k}{k-h+1}$ -- and in particular equal to:

\begin{equation*}
\rho_{EL}(h,k)= \max_{k'\leq k, k'\in \mathcal{N}} 
\frac{k'}{k'-\lfloor h \frac{k'}{k} - \frac{h}{k}\rfloor}
\end{equation*}

The formula for $\rho_{EL}$ is more complex, but vaguely reminiscent of the formula for the ``classic'' $(h,k)$-competitive ratio; and indeed it is easy to verify that for $h=k$ both equal $k$. A detailed analysis of the behaviour of $\rho_{EL}$, including a proof that it is a lower bound on the dynamic $(h,k)-$competitive ratio of any online algorithm, can be found in the following section~\ref{sec:tight}. The remainder of this section is devoted to proving that a dynamic $(h,k)-$competitive ratio $\rho_{EL}(h,k)$ is indeed achieved by all marking algorithms\footnote{A marking algorithm marks a page in memory whenever it accesses it, never evicts a marked page, and unmarks all pages if all are marked and one must be evicted (e.g. in response to a fault or a shrink).} (including MARK, LRU and FWF), by RAND, and by all \emph{dynamically conservative} algorithms. The latter form a class of algorithms that is slightly narrower than that of conservative algorithms\footnote{A conservative algorithm never incurs more than $k$ faults on a sequence of accesses involving at most $k$ distinct pages and a memory of capacity $k$.}\cite{youngconservative} but still includes LRU, FIFO and CLOCK. The cornerstone of the analysis lies in the notion of \emph{short subsequence}, which is the ``correct'' extension of the concept of $k-$phase to dynamic capacity:

\begin{definition}
\label{def:width}
Given a generic (sub)sequence of consecutive requests, its \emph{width} is the number of distinct pages in it. 
\end{definition}

\begin{definition}
\label{def:shortsequence}
Consider a generic request sequence $\sigma$, and a subsequence $\sigma'$ of consecutive requests in $\sigma$ (including page requests, growths and shrinks). $\sigma'$ is \emph{short} if, for every prefix $\pi$ of $\sigma'$, the width of  $\pi$ does not exceed the memory capacity \emph{at the end} of $\pi$.
\end{definition}

\begin{definition}
\label{def:elasticconservative}
A \emph{dynamically conservative} algorithm never incurs more than $w$ faults on any short subsequence of width $w$.
\end{definition}

Note that a dynamically conservative algorithm is also always a conservative algorithm according to the definition of \cite{youngconservative} since with a memory of fixed capacity $k$ every short subsequence involves access to at most $k$ pages, and thus incurs at most $k$ faults. The reverse is not true: LFRU from section~\ref{sec:bad} is conservative but not dynamically conservative. However, we can easily prove:

\begin{theorem}
\label{thm:conservativeset}
LRU, FIFO and CLOCK are dynamically conservative.
\end{theorem}

\begin{proof}
It is not difficult to verify that all three algorithms have the following property: if a page $p$ is brought into memory at time $t$, and a page $p'$ already in memory at time $t$ and is never accessed again, then $p'$ will be evicted before $p$. This holds for LRU because $p$ is more recently accessed than $p'$. It holds for FIFO because $p'$ entered the memory before $p$. It holds for CLOCK because after $t$ the unmark/evict process will encounter $p'$ before encountering $p$ -- thus either evicting or at least unmarking $p'$ before unmarking $p$, and thus certainly evicting it before evicting $p$. Then none of the three algorithms evicts a page accessed during a short sequence before the end of the sequence (since there is always sufficient memory to hold all pages accessed during the sequence), and thus none can incur more faults than the width of the sequence.
\end{proof}

The main result of this section is then:
\begin{theorem}
\label{thm:marking}
The dynamic $(h,k)-$competitive ratio of any online paging algorithm that is either marking or dynamically conservative is no larger than $\rho_{EL}(h,k)$.
\end{theorem}

\begin{proof}
Let us begin with marking algorithms.
The proof bears some resemblance to that of the static case, with a number of subtle but profound differences. One such difference is that, instead of partitioning the request sequence into maximal length phases each involving access to $k$ distinct pages, we partition it into maximal short sequences $\pi_1,\dots,\pi_n$ where $\pi_i$ is the longest short sequence beginning immediately after the end of $\pi_{i-1}$. 

Denote by $w_i$ the width of $\pi_i$. We can assume without loss of generality that the request sequence ends with a page request, so $w_i>0~\forall i$. Note that, for $i>1$, the first request of $\pi_{i}$ must be either a shrink or a request for a page not in $\pi_{i-1}$; in the first case we say that $\pi_{i-1}$ is \emph{capacity bound}, in the second that it is \emph{page bound}.

It is easy to verify by simultaneous induction the following claims hold for all $i$:
\begin{enumerate}
\item All pages in memory are unmarked when $\pi_{i,1}$ is serviced.
\item Every one of the $w_i$ pages accessed during $\pi_i$ (and no other page) remains marked and thus in memory until the end of $\pi_i$.
\item Immediately before $\pi_{i+1,1}$ is serviced, the memory is full and holds $w_i$ pages, all marked.
\end{enumerate}

Claim $1$ holds trivially for $i=1$.
If Claim $1$ holds for $i$, Claim $2$ also holds for $i$, since until the end of $\pi_i$ the memory is large enough to accommodate all pages accessed so far during $\pi_i$, which are the only ones marked.
If Claim $2$ holds for $i$, Claim $3$ also holds for $i$, since immediately before $\pi_{i+1,1}$ is serviced the memory capacity exactly matches the number of distinct pages accessed in $\pi_i$.
If Claim $3$ holds for $i$, Claim $1$ holds for $i+1$ (proving the inductive step), since the first request of $\pi_{i+1}$ must be either a shrink or a request for a page not in $\pi_{i}$, and thus causes all pages in memory to become unmarked.

From Claim $2$ it is obvious that a marking algorithm incurs a number of faults at most equal to $w_i$ during short sequence $\pi_i$, for a total number of faults equal to at most:

\begin{equation}
\label{eqn:elasticonline}
c_{ALG}\leq\sum_{i=1}^n w_i
\end{equation}

Let us compute the number of faults incurred by any other algorithm $\overline{ALG}$ with a memory of capacity at most $\frac{h}{k}$ times that of the marking algorithm, in the interval $\pi'_i$ from immediately after the first request $\sigma$ of $\pi_i$ is serviced, to immediately after the first request of $\pi_{i+1}$ is serviced or to the end of the request sequence if $i=n$. Let $r_i$ be equal to $1$ if $\pi_i$ is page bound, and to $0$ if it is capacity bound, for $1\leq i<n$, and let $r_0=0$ and $r_n=0$. Remember that the first request of a short phase $\pi_i$ is a shrink if $\pi_{i-1}$ is capacity bound, and a page not in $\pi_{i-1}$ if $\pi_{i-1}$ is page bound -- and a growth if $i=1$. Denoting by $w'_i$ the number of distinct pages in $\pi'_i$ after removing the page involved in the first request of $\pi_i$ if any (i.e. if $\pi_{i-1}$ is page bound), we can then write for $1\leq i\leq n$:

\begin{equation}
\label{eqn:distinctpages}
w'_i= -r_{i-1}+w_i+r_{i}
\end{equation}

The subset of these pages in the memory of $\overline{ALG}$ immediately before servicing the first request of $\pi'_{i}$ is then at most:

\begin{equation}
\label{eqn:reservoir}
m_i=
\begin{cases} 
	0 & \text{if $i=1$,}\\
	\lfloor\frac{h}{k}(w_{i-1}-1)\rfloor  &\text{if $i>1$.}
\end{cases}
\end{equation}

Equation~\ref{eqn:reservoir} is immediate if $i=1$ or if $\pi_{i-1}$ is capacity bound - since then the first request of $\pi_{i}$ shrinks the memory available to $ALG$ from $w_{i-1}$ to $w_{i-1}-1$. If instead $\pi_{i-1}$ is page bound, of the $\lfloor\frac{h}{k}w_{i-1}\rfloor$ pages $\overline{ALG}$'s memory can hold, one must be the first page of $\pi_i$ that has just been requested and that does not contribute to $m_i$ -- leaving only $\lfloor\frac{h}{k}w_{i-1}\rfloor-1 \leq \lfloor\frac{h}{k}(w_{i-1}-1)\rfloor$. 
Then the total number of faults incurred by $\overline{ALG}$ is at least:

\begin{equation}
\label{eqn:elasticopt}
\begin{split}
& \sum (w'_i-m_i)\\
        \geq & \sum_{i=1}^n (-r_{i-1}+w_i+r_{i}) 
        - \sum_{i=2}^{n} \lfloor\frac{h}{k}(w_{i-1})-1)\rfloor\\
        \geq & \sum_{i=1}^n (w_i - \lfloor\frac{h}{k}(w_i-1)\rfloor)
\end{split}
\end{equation}

Remembering that both $w_i$ and $w_i - \lfloor\frac{h}{k}(w_i-1)\rfloor$ with $h\leq k$ are positive, and that $\forall a,b,c,d>0$ we have that 
$\frac{a+b}{c+d}=
\frac{c}{c+d}\cdot\frac{a}{c}+\frac{d}{c+d}\cdot\frac{b}{d}
\leq\max(\frac{a}{c},\frac{b}{d})$,
the dynamic $(h,k)-$competitive ratio of $ALG$ is at most:

\begin{equation}
\label{eqn:elasticratio}
\begin{split}
& \frac{\sum_{i=1}^{n} w_i} {\sum_{i=1}^n (w_i - \lfloor\frac{h}{k}(w_i-1)\rfloor)}\\
\leq & \max_i \frac{w_i}{w_i - \lfloor\frac{h}{k}(w_i-1)\rfloor}\\
\leq & \max_{k'\in\{1,\dots,k\}}
\frac{k'}{k' - \lfloor h\frac{k'}{k}-\frac{h}{k}\rfloor}
\end{split}
\end{equation}

This proves the theorem for marking algorithms. The proof for dynamically conservative algorithms proceeds identically, except for the fact that in this case one can immediately obtain, from definition~\ref{def:elasticconservative}, the bound given by Equation~\ref{eqn:elasticonline} on the cost incurred by the online algorithm. 
\end{proof}

The proof of theorem~\ref{thm:marking} is vaguely reminescent of that for marking and conservative algorithms in ``classic'' paging, but is considerably more complex: for example, the strategy of analysing each short subsequence in isolation does not work, and one can only bound the ratio over the whole sequence, through careful accounting and a potential argument. The analysis of RAND faces similar difficulties in terms of ``compartimentalization of costs''; but they are addressed in a different way due to the randomized nature of the algorithm (by exploiting its lack of memory). In this sense it may be somewhat surprising that exactly the same bound obtained in theorem~\ref{thm:marking} also applies to RAND:

\begin{theorem}
\label{thm:elasticrand}
RAND's dynamic $(h,k)-$competitive ratio is no larger than $\rho_{EL}(h,k)$ in the adaptive online adversary model.
\end{theorem}
\begin{proof}
While the proof bears a few similarities to the analysis in the static case, it requires subtlety and a somewhat different approach due to the possible fluctuations of memory capacity -- and  in particular to the fact that, if $h < k$, cache shrinks may not be ``synchronized'' and RAND may incur shrinks when the optimal offline algorithm OPT does not. Instead of comparing the number of faults $c_{RAND}$ and $c_{OPT}$ incurred, respectively, by RAND and OPT, we then begin by comparing the number of \emph{evictions} $e_{RAND}$ and $e_{OPT}$. For simplicity, assume that, after any given request (for a page, or for a capacity change), the request is served in the following order. OPT performs any eviction(s); then it loads into memory any requested page not yet there; then RAND does the same; finally, OPT adjusts its memory capacity, and then RAND does the same. 

Let the \emph{garbage} of RAND at any given point in time be the set $G$ of pages in its memory and not in the memory of OPT. First of all, note that $G$ can increase only when OPT incurs an eviction (and at most by $1$ page for each eviction), since RAND never brings into memory a page not requested by OPT -- which at that point must then be in OPT's memory. 

Immediately before RAND incurs an eviction, its memory must be full; denote by $k'$ and $h'$ the memory capacity of RAND and OPT at that point. If the eviction is the result of a shrink, then the number of pages in RAND's memory that are \emph{not} garbage are at most:
\begin{equation}
\label{eqn:randshrink}
h'=\lfloor\frac{h}{k}(k'-1)\rfloor
\end{equation}

Note that at this point OPT has adjusted its memory capacity to the shrink but RAND has not. If the eviction is the result of a page fault, then the requested page at this point is in OPT's memory but not in RAND's, and the number of pages in RAND's memory that are \emph{not} garbage are at most:

\begin{equation}
\label{eqn:randfault}
h'-1=\lfloor\frac{h}{k}k'\rfloor-1\leq \lfloor\frac{h}{k}(k'-1)\rfloor
\end{equation}

Thus the probability that, when RAND incurs an eviction, $|G|$ decreases by $1$ is at least:
\begin{equation}
\label{eqn:randprob}
p_{RAND}=\min_{k'\in\{1,\dots,k\}} \frac{k'-\lfloor\frac{h}{k}(k'-1)\rfloor}{k'}
\end{equation}

and at any given time we have that, in expectation:

\begin{equation}
\label{eqn:randpotential}
|G|\leq e_{OPT}- e_{RAND}\cdot p_{RAND}
\end{equation}

Appending to any request sequence sufficient shrinks to bring $RAND$'s memory capacity to $0$ obviously brings $|G|$ to $0$, without increasing the number of \emph{faults} incurred by $RAND$ or $OPT$. For any algorithm that evicts a single page at a time, when the memory holds no pages the number of faults and evictions incurred must coincide. Setting $|G|$ to $0$, as well as $c_{OPT}=e_{OPT}$ and $c_{RAND}=e_{RAND}$, in Equation~\ref{eqn:randpotential} then yields for RAND a dynamic $(h,k)-$competitive ratio equal to at most:

\begin{equation}
\label{eqn:randratio}
\frac{c_{RAND}}{c_{OPT}}\leq \frac{1}{p_{RAND}} = 
\max_{k'\in\{1,\dots,k\}} \frac{k'}{k'-\lfloor h\frac{k'}{k} - \frac{h}{k}\rfloor}
\end{equation}
\end{proof}

\vspace*{3mm}
\section{An exact characterization of the competitive ratio}
\label{sec:tight}
\vspace*{1mm}

The upper bound $\rho_{EL}(h,k)$ obtained in the previous section~\ref{sec:good} for the dynamic $(h,k)-$competitive ratio of many online paging algorithms  is actually tight. More formally, we can prove:

\begin{theorem}
\label{thm:elasticcompetitivelb}
No online paging algorithm has a dynamic $(h,k)-$competitive ratio (against any online of offline adaptive adversary if randomized) lower than:
\begin{equation*}
\rho_{EL}(h,k)= max_{k'\in\{1,\dots,k\}} 
\frac{k'}{k'-\lfloor h \frac{k'}{k} - \frac{h}{k}\rfloor}
\end{equation*}
\end{theorem}

\begin{proof}
Let $\overline{k} = \argmax_{k'\in\{1,\dots,k\}} 
\frac{k'}{k'-\lfloor h \frac{k'}{k} - \frac{h}{k}\rfloor}$, 
and let $ALG$ be a generic online paging algorithm.
Consider a request sequence 
$\sigma_n=<+^{(\overline{k}-1)},\pi_1,\dots,\pi_n>$, where:

\begin{equation}
\label{eqn:onlineworstsubsequence}
\pi_i=<+^{(k-\overline{k}+1)},p_{i,1},-^{(k-\overline{k}+1)}, \dots, 
+^{(k-\overline{k}+1)}, p_{i,\overline{k}}, -^{(k-\overline{k}+1)}>
\end{equation}
and $p_{i,j}$ is any one page, from the set $p_1,\dots,p_{\overline{k}}$, that is not in $ALG$'s memory just before it is requested -- note that immediately before any page request ALG's memory holds at most $\overline{k}-1$ pages, so there always exists one such page. $ALG$ then incurs a fault on every page request, for a total number of faults equal to:

\begin{equation}
\label{eqn:onlinelbonlinecost}
c_{ALG}(\sigma_n) = n\cdot\overline{k}
\end{equation} 

Consider an offline algorithm $\overline{ALG}$ with access to a memory that has at most $\frac{h}{k}$ times the capacity of $ALG$'s at any given time; in particular, $\overline{ALG}$'s memory capacity grows to $h$ immediately before any page request, and immediately afterwards drops to capacity: 

\begin{equation}
\label{eqn:onlinelbofflinesize}
\lfloor\frac{h}{k}(\overline{k}-1)\rfloor=
\lfloor h \frac{\overline{k}}{k} - \frac{h}{k}\rfloor < h
\end{equation}

$\overline{ALG}$ can easily maintain in its ``permanent'' 
$\lfloor h \frac{\overline{k}}{k} - \frac{h}{k}\rfloor$ memory locations the 
$\lfloor h \frac{\overline{k}}{k} - \frac{h}{k}\rfloor$ pages with most \emph{expected} accesses in $\sigma_n$, incurring for each only one initial fault.
Note that the total number of accesses to these pages is, in expectation, at least
$n\overline{k}\cdot\frac{\lfloor h \frac{\overline{k}}{k} - \frac{h}{k}\rfloor}{\overline{k}}=n \lfloor h \frac{\overline{k}}{k} - \frac{h}{k}\rfloor$. Every other page, when requested, is brought into the ``temporary'' location(s) immediately eliminated by the following shrink. $\overline{ALG}$ then incurs an expected number of faults equal to:

\begin{equation}
\label{eqn:onlinelbofflinecost}
c_{\overline{ALG}}(\sigma_n)\leq\lfloor h \frac{\overline{k}}{k} - \frac{h}{k}\rfloor+
n(\overline{k}-\lfloor h \frac{\overline{k}}{k} - \frac{h}{k}\rfloor)
\end{equation}

Then the competitive ratio of $ALG$ can be no lower than:

\begin{equation}
\label{onlineratiolimit}
\begin{split}
lim_{n\rightarrow\infty}\frac{c_{ALG}(\sigma_n)}{c_{\overline{ALG}}(\sigma_n)}=
& lim_{n\rightarrow\infty} \frac{n\overline{k}}
{\lfloor h \frac{\overline{k}}{k} - \frac{h}{k}\rfloor+
n(\overline{k}-\lfloor h \frac{\overline{k}}{k} - \frac{h}{k}\rfloor)}\\
& = \max_{k'\in\{1,\dots,k\}}\frac{k'}
{k'-\lfloor h \frac{k'}{k} - \frac{h}{k}\rfloor}
\end{split}
\end{equation}

It is important to observe that, if $ALG$ is randomized, $\overline{ALG}$ need only know $ALG$'s probabilistic behaviour to choose which pages to keep in its own memory; and it can choose which page to request next based only on $ALG$'s current memory contents. Thus the lower bound we proved holds for deterministic and randomized algorithms both in the adaptive offline and in the adaptive online adversary models. 
\end{proof}

As noted in section~\ref{sec:good} the expression of the optimal dynamic 
$(h,k)$-compet-itive ratio $\rho_{EL}(h,k)$
appears considerably more complex than, but vaguely reminiscent of, that of the ``classic'' bound on the $(h,k)-$competitive ratio, $\frac{k}{k-h+1}$. It is natural to ask whether the two are actually different, and if so to what extent. We show that $\rho_{EL}(h,k)$ \emph{is}, in fact, a factor $\approx 1+\frac{1}{k}$ larger for some ``natural'' values of $h$ and $k$ -- though it is never more than a factor $1+\frac{1}{k}$ larger, and actually coincides with $\frac{k}{k-h+1}$ if $h$ is equal or very close to $k$. This is stated more formally in the following two theorems:

\begin{theorem}
\label{thm:onlinelbstrict}
For any odd $h$ and $k=2h$, 
$\rho_{EL}(h,k)\geq (1+\frac{1}{k}-\frac{2}{k^2}) \frac{k}{k-h+1}$.
\end{theorem}
\begin{proof}
For any integer $i\geq 0$, choosing $h=2i+1$, $k=2h$, and $k'=k-1$, we obtain immediately:
\begin{equation}
\label{eqn:onlinelbstrict}
\begin{split}
& \rho_{EL}(h,k) \geq \frac{4i+1}
{4i+1 - \lfloor (2i+1) \frac{4i+1}{4i+2} - \frac{2i+1}{4i+2}\rfloor}\\
& =\frac{4i+1}{4i+1-\lfloor\frac{4i+1}{2} - \frac{1}{2}\rfloor}
= \frac{4i+1}{2i+1} 
= \frac{k-1}{\frac{k}{2}}\\
& = \frac{k-1}{\frac{k}{2}} \cdot \frac{\frac{k}{2}+1}{k} \cdot \frac{k}{k-h+1}\\
& = (1+\frac{1}{k}-\frac{2}{k^2}) \frac{k}{k-h+1}
\end{split}
\end{equation}
\end{proof}

\begin{theorem}
\label{thm:roelub}
$\frac{k}{k-h+1} \leq \rho_{EL}(h,k)< (1+\frac{1}{k})\frac{k}{k-h+1}$ for all $h,k\in\mathbb{Z}^+$ with $h\leq k$, and $\rho_{EL}(h,k) = \frac{k}{k-h+1}$ if $k\geq h > k-\sqrt{k}$.
\end{theorem}
\begin{proof}
It is immediate to verify that, for $k'=k$:
\begin{equation}
\label{eqn:roelvscompratio}
\rho_{EL}(h,k)\geq \frac{k'}{k'-\lfloor h \frac{k'}{k} - \frac{h}{k}\rfloor} = \frac{k}{k-h+1}
\end{equation}

And since, if $k'\leq h$, we have that: 
\begin{equation}
\label{eqn:roellimitedrange}
\begin{split}
\frac{k'}{k'-\lfloor h \frac{k'}{k} - \frac{h}{k}\rfloor}
\leq \frac{k'}{k'- (h \frac{k'}{k} - \frac{h}{k})} \\
= \frac{k'\frac{k}{k'}}
{k'\frac{k}{k'} -  h \frac{k'}{k}\frac{k}{k'} +  \frac{h}{k}\frac{k}{k'}}
\leq \frac{k}{k-h+1}
\end{split}
\end{equation} 

then values of $k'\leq h$ can be disregarded in the $\max$ operation. 
To prove that, for all $h \leq k$, 
$\rho_{EL}(h,k)\leq (1+\frac{1}{k})\frac{k}{k-h+1}$,
note that:

\begin{equation}
\label{eqn:roelsmooth}
\begin{split}
& \rho_{EL}(h,k)
= max_{k'\in\{1,\dots,k\}} 
\frac{k'}{k'-\lfloor h \frac{k'}{k} - \frac{h}{k}\rfloor}\\
& \leq max_{k'\in\{1,\dots,k\}} 
\frac{k'}{k'- (h \frac{k'}{k} - \frac{h}{k})}\\
& = \frac{k}{k-(h \frac{k}{k} - \frac{h}{k})}
= \frac{k}{k-h +\frac{h}{k}}
\end{split}
\end{equation}

Then we obtain:

\begin{equation}
\label{eqn:roelsmall}
\frac{\rho_{EL}(h,k)}{\frac{k}{k-h+1}}
\leq \frac{k-h+1}{k-h+\frac{h}{k}}
= 1+\frac{\frac{k-h}{k}}{k-h+\frac{h}{k}}
= 1+\frac{1}{k+\frac{h}{k-h}}
< 1+\frac{1}{k}
\end{equation}

To prove that $\rho_{EL}(h,k)$ coincides with $\frac{k}{k-h+1}$ for $h>k-\sqrt{k}$
let us rewrite $h$ and $k'$ as $h=k-a$ and $k'=k-b$, 
with  $a>b$ and $a,b\in\mathbb{Z}_0^+$. We obtain:

\begin{equation} 
\label{eqn:roelinteger}
\begin{split}
& \rho_{EL}(h>(k-\sqrt{k}),k)
= \max_{\sqrt{k}>a>b} 
\frac{k-b}{k-b-\lfloor  \frac{(k-a)(k-b)}{k} - \frac{k-a}{k}\rfloor}\\
& = \max_{\sqrt{k}>a>b}
\frac{k-b}{k-b-\lfloor k-(a+b)+\frac{ab}{k} - 1 + \frac{a}{k}\rfloor}\\
& = \max_{\sqrt{k}>a>b}
\frac{k-b}{k-b-k+(a+b)+1-\lfloor \frac{a(b+1)}{k}\rfloor}\\
& < \max_{\sqrt{k}>a>b}
\frac{k}{k-h+1-\lfloor\frac{a(b+1)}{k}\rfloor}
 = \frac{k}{k-h+1}
\end{split}
\end{equation} 

where the last equality follows from the fact that, since $a=k-h<\sqrt{k}$ and $b\leq a-1<\sqrt{k}-1$, then $a(b+1)<k$. 
\end{proof}

The complex expression of $\rho_{EL}(h,k)$ is in part due to the ``rounding'' of the memory capacity of the optimal offline algorithm. However, it is important to note that this rounding is not sufficient to explain why $\rho_{EL}(h,k)$ can be strictly \emph{larger} than the ``classic'' ratio $\frac{k}{k-h+1}$ obtained when capacity is fixed at its maximum value: at smaller capacities rounding can only favour the online algorithm, and for any fixed ratio $\frac{k'}{h'}$, $\frac{k'}{k'-h'+1}$ strictly decreases with $k'$, again favouring the online algorithm at smaller capacities. \emph{Capacity fluctuations (rather than simply the choice between different, constant capacities) are then the source of the separation} between $\rho_{EL}(h,k)$ and the ``classic'' $(h,k)$-competitive ratio $\frac{k}{k-h+1}$.

\vspace*{3mm}
\section{Decoupling replacement from capacity in RAM rental}
\label{sec:applications}
\vspace*{1mm}

The results from section \ref{sec:good} can be readily applied to the RAM rental problem, in which a paging algorithm ALG can choose the capacity sequence (with maximum capacity $k$), and the cost it incurs and must minimize on a request sequence $\sigma$ is:
\begin{equation}
\label{eqn:ramrental}
R_{ALG}^k(\sigma) = \sum_{i=1}^{|\sigma|} (\alpha f(i) + \beta w(i))
\end{equation}
where $w(i)$ is the capacity when serving the $i^{th}$ request of $\sigma$, and $f(i)$ is $1$ if that request is a fault and $0$ otherwise. 
The fundamental consequence of our results from section~\ref{sec:good} is that to a large extent the replacement policy can be decoupled from the choice of capacities. More precisely, theorem~\ref{thm:marking} yields:

\begin{corollary}
\label{cor:decoupling}
Consider a paging algorithm $ALG$, servicing each request $\sigma_i$ of a sequence $\sigma$ with capacity $w(i)\leq h$ and an arbitrary (even offline) replacement policy; and a second paging algorithm $ALG'$ servicing $\sigma_i$ with capacity $2w(i)$ and a replacement policy that can be any marking or dynamically conservative algorithm. Then, for any choice of $\alpha$, $\beta$ and $w(\cdot)\leq h$:
\begin{equation}
\label{eqn:decoupling}
R_{ALG'}^{2h}(\sigma) \leq 2\cdot R_{ALG}^h(\sigma).
\end{equation}
\end{corollary}
which follows immediately from the fact that the sum of all faults incurred by $ALG'$ is at most twice that by $ALG$ as long as $ALG'$ maintains twice the capacity of $ALG$.  \emph{In other words, RAM rental is all about choosing the correct capacity at any given time; and any of the ``classic'' replacement policies analysed in the previous section will be close to optimal for any choice of $\alpha$, of $\beta$, and of the capacity sequence.}

\vspace*{3mm}
\section{Conclusions}
\label{sec:end}
\vspace*{1mm}

Good performance in the case of constant memory capacity provides no performance guarantees whatsoever in the case of fluctuating memory capacity: moving from a scenario where capacity remains constant to one where it can fluctuate by a single page can mean the difference between performance optimal within a factor $2$, and performance suboptimal by an arbitrarily large factor. This suggests the need of extreme caution when evaluating with classic methodologies the performance of paging algorithms meant for memory systems with dynamic capacity.

A counterpoint to this very ``negative'' result is that several extremely simple classic paging algorithms achieve optimal or nearly optimal performance even in the dynamic capacity framework. This is particularly surprising because none of these algorithms is designed to take memory capacity fluctuations into explicit account: counterintuitively, while knowledge of future page requests provides an advantage, knowledge of future memory capacity does not. A practical corollary is that, in the design of memory architectures, one can then efficiently decouple the problem of allocating memory resources to different cores/processes/threads from the problem of managing the allocated memory -- greatly simplifying system design and analysis and providing a strong (a posteriori!) theoretical justification for the exokernel approach~\cite{exokernel}.

As in classic paging, in the dynamic capacity framework competitive analysis fails to distinguish between the performance of LRU, of FIFO, and of more naive algorithms such as RAND or FWF -- at least without resorting to more sophisticated approaches such as access graphs. While each of these algorithms is still guaranteed to outperform an optimal offline algorithm (and thus any other online algorithm) whose memory system has half the capacity and twice the access cost, there are probably differences within those factors of $2$ that would be important to characterize in practice. It is by no means clear whether the winner in the dynamic capacity scenario would be the same as in the classic one, or whether models designed a posteriori to explain the superiority of e.g. LRU over FIFO would still provide correct predictions. 

In this sense we are not aware of any experimental benchmarks specifically designed to assess the impact of memory capacity fluctuations. A fundamental obstacle in their development seems to be the difficulty of characterizing ``typical'' fluctuation patterns encountered in practice. An interesting line of inquiry would be to investigate whether one can obtain, from the performance numbers of a black box algorithm under a small ``basis'' of specific fluctuation patterns, sufficient information to compute a good assessment of the algorithm's performance numbers under any other pattern.

Finally, the ``dynamic resources'' approach is not necessarily restricted to paging. There are a number of other problems where the amount of resources available for a task can realistically vary over time. Examples include call admission \cite{calladmission1} (with variable circuit capacity) and the numerous variants of online scheduling \cite{schedulingsurvey1} (with e.g. variable number or speed of servers). In addition to studying each problem individually, it would be extremely interesting to identify broad classes sharing similar characteristics. For example, which problems can be solved optimally or almost optimally without knowledge of the amount of resources available in the future (as in the case of paging with dynamic memory capacity)?


\begin{thebibliography}{10}

\bibitem{multifinger}
A.Fiat and A.~Karlin.
\newblock Randomized and multipointer paging with locality of reference.
\newblock In {\em Proc. ACM STOC}, pages 626--634, 1995.

\bibitem{averageorder}
S.~Albers, L.~M. Favrholdt, and O.~Giel.
\newblock On paging with locality of reference.
\newblock {\em J. Comput. Syst. Sci.}, 70(2):145--175, 2005.

\bibitem{bijective}
S.~Angelopoulos, R.~Dorrigiv, and A.~L{\'o}pez-Ortiz.
\newblock On the separation and equivalence of paging strategies.
\newblock In {\em Proc. ACM/SIAM SODA}, pages 229--237, 2007.

\bibitem{belady}
L.~Belady.
\newblock A study of replacement algorithms for a virtual-storage computer.
\newblock {\em IBM Systems}, 5(2):78--101, 1966.

\bibitem{maxmax}
S.~Ben-David and A.~Borodin.
\newblock A new measure for the study of on-line algorithms.
\newblock {\em Algorithmica}, 11(1):73--91, 1994.

\bibitem{online}
A.~Borodin and R.~El-Yaniv.
\newblock {\em Online Computation and Competitive Analysis}.
\newblock {Cambridge University Press}, 1998.

\bibitem{accessgraph}
A.~Borodin, P.~Raghavan, S.~Irani, and B.~Schieber.
\newblock Competitive paging with locality of reference.
\newblock In {\em Proc. ACM STOC}, pages 249--259, 1991.

\bibitem{competitivealternatives}
J.~Boyar, S.~Irani, and K.~S. Larsen.
\newblock A comparison of performance measures for online algorithms.
\newblock In {\em Proc. WADS}, 2009.

\bibitem{ramrental0}
M.~Chrobak.
\newblock Sigact news online algorithms column 17.
\newblock {\em SIGACT News}, 41(4):114--121, 2010.

\bibitem{onlinesurvey05}
R.~Dorrigiv and A.~L{\'o}pez-Ortiz.
\newblock A survey of performance measures for on-line algorithms.
\newblock {\em SIGACT News}, 36(3):67--81, 2005.

\bibitem{exokernel}
D.~R. Engler, M.~F. Kaashoek, and J.~O'Toole.
\newblock Exokernel: An operating system architecture for application-level
  resource management.
\newblock In {\em Proc. SOSP}, pages 251--266, 1995.

\bibitem{marking}
A.~Fiat, R.~Karp, M.~Luby, L.~McGeoch, D.~Sleator, and N.~Young.
\newblock Competitive paging algorithms.
\newblock {\em Journal of Algorithms}, 12(4):685--699, 1991.

\bibitem{cacheoblivious}
M.~Frigo, C.~E. Leiserson, H.~Prokop, and S.~Ramachandran.
\newblock Cache-oblivious algorithms.
\newblock In {\em Proc. IEEE FOCS}, 1999.

\bibitem{calladmission1}
J.~A. Garay, I.~S. Gopal, S.~Kutten, Y.~Mansour, and M.~Yung.
\newblock Efficient on-line call control algorithms.
\newblock {\em Journal of Algorithms}, 23(1):180--194, 1997.

\bibitem{multicore}
A.~Hassidim.
\newblock Cache replacement policies for multicore processors.
\newblock In {\em Proc. ICS}, 2010.

\bibitem{sleepycache}
H.~Homayoun, M.~Makhzan, and A.~Veidenbaum.
\newblock Multiple sleep mode leakage control for cache peripheral circuits in
  embedded processors.
\newblock In {\em Proc. CASES}, pages 197--206, 2008.

\bibitem{iranicachechanges}
S.~Irani.
\newblock Page replacement with multi-size pages and applications to web
  caching.
\newblock {\em Algorithmica}, 33(3):384--409, 2002.

\bibitem{competitive}
A.~R. Karlin, M.~S. Manasse, L.~Rudolph, and D.~D. Sleator.
\newblock Competitive snoopy caching.
\newblock In {\em Proc. SFCS}, 1986.

\bibitem{diffuseadversary}
E.~Koutsoupias and C.~H. Papadimitriou.
\newblock Beyond competitive analysis.
\newblock {\em SIAM Journal on Computing}, 30(1):300--317, 2000.

\bibitem{lrfu}
D.~Lee, J.~Choi, J.~hun Kim, S.~H. Noh, S.~L. Min, Y.~Cho, and C.~S. Kim.
\newblock {LRFU}: A spectrum of policies that subsumes the least recently used
  and least frequently used policies.
\newblock In {\em Proc. ACM SIGMETRICS}, 1999.

\bibitem{ramrental1}
A.~Lopez-Ortiz and A.~Salinger.
\newblock Minimizing cache usage in paging.
\newblock In {\em Proc. WAOA}, 2012.

\bibitem{arc}
N.~Megiddo and D.~S. Modha.
\newblock Outperforming {LRU} with an adaptive replacement cache algorithm.
\newblock {\em Computer}, 37:58--65, April 2004.

\bibitem{schedulingsurvey1}
K.~Pruhs.
\newblock Competitive online scheduling for server systems.
\newblock {\em SIGMETRICS Perform. Eval. Rev.}, 34(4):52--58, Mar. 2007.

\bibitem{utilitypartition}
M.~K. Qureshi and Y.~N. Patt.
\newblock Utility-based cache partitioning: A low-overhead, high-performance,
  runtime mechanism to partition shared caches.
\newblock In {\em Proc. IEEE/ACM MICRO}, 2006.

\bibitem{rand}
P.~Raghavan and M.~Snir.
\newblock Memory versus randomization in on-line algorithms.
\newblock In {\em Proc. ICALP}, 1989.

\bibitem{lrufifo}
D.~Sleator and R.~Tarjan.
\newblock Amortized efficiency of list update and paging rules.
\newblock {\em Communications of the ACM}, 28(2):202--208, 1985.

\bibitem{nonfault}
E.~Torng.
\newblock A unified analysis of paging and caching.
\newblock {\em Algorithmica}, 20(2):175--200, 1998.

\bibitem{youngconservative}
N.~Young.
\newblock The k-server dual and loose competitiveness for paging.
\newblock {\em Algorithmica}, 11:525--541, 1994.

\bibitem{youngcachechanges}
N.~E. Young.
\newblock On-line caching as cache size varies.
\newblock In {\em Proc. ACM/SIAM SODA}, 1991.

\end{thebibliography}
\end{document}